\documentclass[preprintnumbers,amsmath,amssymb]{revtex4}

\usepackage{mathrsfs}
\usepackage{amssymb}
\usepackage{graphicx}
\usepackage{dcolumn}
\usepackage{bm}
\usepackage{textcomp}
\usepackage{amsmath}
\usepackage[titletoc]{appendix}
\usepackage{accents}
\usepackage{ntheorem}

\newtheorem{lemma}{Lemma}

\newtheorem{theorem}{Theorem}

\newtheorem*{proof}{Proof}

\newcommand\ket[1]{\ensuremath{|#1\rangle}}
\newcommand\bra[1]{\ensuremath{\langle#1|}}

\newcommand\oprod[2]{\ensuremath{|#1\rangle\langle#2|}}
\newcommand\mean[1]{\ensuremath{\langle #1\rangle}}
\newcommand\tr{\mathop{\rm tr}\nolimits}
\newcounter{RomanNumber}

\makeatletter
\def\widebar{\accentset{{\cc@style\underline{\mskip10mu}}}}
\def\Widebar{\accentset{{\cc@style\underline{\mskip8mu}}}}
\makeatother

\begin{document}
\title{Zigzag approach to higher key rate of sending-or-not-sending twin field quantum key distribution with finite key effects}
\author{Cong Jiang$ ^{1}$, Xiao-Long Hu$ ^{1}$, Hai Xu$^1$, Zong-Wen Yu$ ^{1,3}$
and Xiang-Bin Wang$ ^{1,2,4,5\footnote{Email Address: xbwang@mail.tsinghua.edu.cn}\footnote{Also at Center for Atomic and Molecular Nanosciences, Tsinghua University, Beijing 100084, China}}$}

\affiliation{ \centerline{$^{1}$State Key Laboratory of Low
Dimensional Quantum Physics, Department of Physics,} \centerline{Tsinghua University, Beijing 100084, China}
\centerline{$^{2}$ Synergetic Innovation Center of Quantum Information and Quantum Physics,}\centerline{University of Science and Technology of China, Hefei, Anhui 230026, China
 }
\centerline{$^{3}$Data Communication Science and Technology Research Institute, Beijing 100191, China}
\centerline{$^{4}$ Jinan Institute of Quantum technology, SAICT, Jinan 250101, China}
\centerline{$^{5}$ Shenzhen Institute for Quantum Science and Engineering, and Physics Department,}
\centerline{Southern University of Science and Technology, Shenzhen 518055, China}}
\begin{abstract}
Odd-parity error rejection (OPER) can drastically improve the asymptotic key rate of sending-or-not-sending twin-field (SNS-TF) quantum key distribution (QKD). However, in practice, the finite key effects have to be considered for security. Here, we propose a zigzag approach to verify the phase-flip error of the survived bits after OPER. Based on this, we can take all the finite key effect efficiently in calculating the non-asymptotic key rate. Numerical simulation shows that our method here produces the  highest key rate over all distances among all existing methods, improving the key rate by more than $100\%$ to $3000\%$  in comparison with different prior art methods with typical experimental setting. Also, we show that with the method here, the SNS-TF QKD can by far break the the absolute bound of repeater-less key rate with whatever detection efficiency. We can even reach a non-asymptotic key rate more than $40$ times of the practical bound and $13$ times of the absolute bound with $10^{12}$ pulses. Besides, we apply the McDiarmid inequality to estimate the phase flip error rate, further improving the key rate by more than $20\%$.
\end{abstract}


\maketitle
\section{Introduction}
Quantum key distribution (QKD)\cite{bennett1984quantum,gisin2002quantum,gisin2007quantum,
scarani2009security,shor2000simple,koashi2009simple,tamaki2003unconditionally,kraus2005lower}  can provide secure private communication between two remote parties Alice and Bob~. In the recent years, the efficiency and security of QKD in practice have been extensively studied, e.g., very recently, the important idea  named in Twin-Field (TF) QKD~\cite{lu2018overcoming}, and its variants~\cite{wang2018twin,tamaki2018information,ma2018phase,lin2018simple,cui2019twin,curty2018simple,yu2019sending,maeda2019repeaterless,
lu2019twin,jiang2019unconditional,xu2019general,hu2019general,zhang2019twin,zhou2019asymmetric}. 

The security proof under ideal conditions does not directly apply to the practical QKD where imperfect devices such as the weak coherent state (WCS) sources or the imperfect detectors such as avalanche photodiode detectors (APDs) are used~\cite{huttner1995quantum,yuen1996quantum,brassard2000limitations,lu2000security,lu2002quantum,lydersen2010hacking,
gerhardt2011full,scarani2008quantum}. The decoy-state method~\cite{hwang2003quantum,wang2005beating,lo2005decoy} can assure the security of the QKD protocol with imperfect sources and maintain the high key rate, and thus attracts many studies on both  theories~\cite{wang2007quantum,ad2007simple,wang2007simple,wang2008general,wang2009decoy,
yu2016reexamination,chau2018decoy} and experiments~\cite{rosenberg2007long,schmitt2007experimental,peng2007experimental,liao2017satellite,peev2009secoqc,chen2010metropolitan,
sasaki2011field,frohlich2013quantum,boaron2018secure,wang2008experimental,xu2009experimental}. Besides decoy-state method, there are other protocols such as RRDPS protocol~\cite{sasaki2014practical,takesue2015experimental} proposed to beat photon-number-splitting (PNS) attack. Measurement-Device-Independent (MDI)-QKD~\cite{braunstein2012side,lo2012measurement} was proposed to solve all possible detection loopholes. The decoy-state MDI-QKD can assure the security with imperfect sources and detectors, and thus has been widely studied~\cite{
wang2013three,rubenok2013real,liu2013experimental,tang2014experimental,tang2014measurement,wang2015phase,comandar2016quantum,
yin2016measurement,wang2017measurement,xu2013practical,curty2014finite,xu2014protocol,
yu2015statistical,zhou2016making,jiang2017measurement}.

The very recently proposed TF QKD~\cite{lu2018overcoming}, together with its variants~\cite{wang2018twin,tamaki2018information,ma2018phase,lin2018simple,cui2019twin,curty2018simple,yu2019sending,maeda2019repeaterless,
lu2019twin,jiang2019unconditional,xu2019general,hu2019general,zhang2019twin,zhou2019asymmetric}, changes the key rate into square root scale of channel transmittance. The protocol can break the repeater-less key rate limit such as the TGW bound~\cite{takeoka2014fundamental} presented by Takeoka, Guha, and Wilde, and the PLOB bound~\cite{pirandola2009direct} established by Pirandola, Laurenza, Ottaviani, and Banchi, which will be used in this work. The well-known PLOB bound also corresponds to the secret key capacity of the lossy communication channel. Experiments~\cite{minder2019experimental,liu2019experimental,wang2019beating,zhong2019proof,fang2019surpassing} have been done to demonstrate those protocols. In particular, the efficient variant of TF QKD, named in sending-or-not-sending (SNS) protocol has been proposed in Ref.~\cite{wang2018twin}. The SNS protocol has its advantage of tolerating large misalignment errors~\cite{wang2018twin,yu2019sending} and unconditional security with finite pulses~\cite{jiang2019unconditional}. The numerical results show that the secure distance can exceed $500$ km even when the misalignment error is as large as $20\%$~\cite{jiang2019unconditional}. The SNS protocol has been experimentally demonstrated in proof-of-principle in Ref.~\cite{minder2019experimental}, and realized in real optical fiber with the finite key effects being taken into consideration~\cite{liu2019experimental}.

However, there are still considerable spaces to further improve the performance of the SNS protocol. For example, the original SNS protocol~\cite{wang2018twin,yu2019sending,jiang2019unconditional} is limited to small probability of sending a signal coherent state and this limits its key rate.

The two-way communication method, with odd-parity error rejection (OPER) can be applied to SNS protocol and thus increases the probability of sending the signal coherent state~\cite{xu2019general}. This method can improve the asymptotic key rate and secure distance of SNS protocol drastically.  However, there are only finite pulses and finite intensities in practice and we need to consider the finite key effects. In this article, we study the the finite key effects for the SNS protocol with OPER. Extensive comparisons among different protocols show that the key rate simulated by the method of this work presents the highest key rate known so far if we assume a reasonable finite size of the key. It presents advantageous results by more than $2$ to $30$ times of the prior art results.  Besides, we apply the improved version of McDiarmid inequality~\cite{chau2019application} to the estimating of phase-flip error rate, which can reduce the effect of statistical fluctuation and further improve the key rate especially when the key size is small. Also, we show that with the method here, the SNS-TF QKD can by far break the absolute key rate limit of the repeater-less QKD. We can even reach a non-asymptotic key rate more than $13$ times of this absolute limit.

Breaking the repeater-less key-rate limit is the most fascinating property of the TF-QKD. In practice, there are lot of constraints. We set the following conditions for practical QKD results that unconditionally break the repeater-less key rate limit:

\noindent \textbf{1}. The protocol itself must be secure.

\noindent \textbf{2}. A reasonable finite size of the key should be assumed and the finite key effects should be considered in calculation.

\noindent \textbf{3}. The actual key rate should break the absolute PLOB bound~\cite{pirandola2017fundamental}, which upper bounds the repeater-less key rate given whatever local devices, including the perfect detection devices.

For this goal, we need to consider the finite key effects and we need to show a high key rate exceeding the absolute PLOB bound.

Before going into further details, let us first briefly review the OPER and see why the problem of finite key effects is not straightforward.
To improve the performance, we can apply the bit-flip error rejection in the SNS protocol. Say, Alice and Bob group their raw bits two by two. To each group of two-bits, they compare the parity values. They discard both bits of the group when they have different parity values for the group and they discard one bit and keep another bit in the group if they have the same parity values for the group. After error rejection, the bit-flip rate is expected to be reduced but the phase-flip error rate is normally increased. Moreover, due to the structure of SNS protocol, the bit-flip errors rise if we raise the sending probability. As was shown in \cite{xu2019general}, it will be more efficient in reducing bit-flip errors if we only use odd parity events in the error rejection, and we name this as odd-parity error rejection (OPER).   In the SNS protocol, all bit-flip errors come from tagged bits. There is no bit-flip error for those untagged bits, so the phase-flip error rate is iterated by the following formula~\cite{chau2002practical,gottesman2003proof}:
\begin{equation}
e_1^\prime=2e_1(1-e_1),
\end{equation}
where $e_1$ is the phase-flip error rate of those raw untagged bits before error rejection and  $e_1^\prime$ is the phase-flip error rate for those survived untagged bits after error rejection.

Note that this iteration formula is the averaged value of phase-flip error rate for the survived untagged bits from both odd-parity groups and even parity groups. However, we cannot blindly use this iteration formula for the phase-flip error rate for survived untagged bits from odd-parity groups alone, because the phase-flip error rates of survived untagged bits from odd parity is in general {\em different} from that of even-parity groups. For example, consider the virtual protocol that Alice and Bob share many entangled pair states with each pair in the state $|\Psi\rangle=\frac{1}{\sqrt{2}}(|00\rangle+e^{i\phi}|11\rangle)$, where $0$ and $1$ mean qubit $0$ and $1$, and the first qubit belongs to Alice and the second qubit belongs to Bob. The random grouping of two untested bits in the real protocol is now related to the random grouping of two pair states in the virtual protocol. According to Ref.~\cite{kraus2007security}, the operators of OPER are
\begin{equation}\label{opero}
\begin{split}
\hat{O}_A=&\ket{0}_A\bra{0}_A\bra{1}_A+\ket{1}_A\bra{1}_A\bra{0}_A,\\
\hat{O}_B=&\ket{0}_B\bra{0}_B\bra{1}_B+\ket{1}_B\bra{1}_B\bra{0}_B.
\end{split}
\end{equation} 
The subscript $A$ means it acts on Alice's qubit and subscript $B$ means it acts on Bob's qubit. And the operators of even-parity error rejection are 
\begin{equation}
\begin{split}
\hat{O}_A^\prime=&\ket{0^\prime}_A\bra{0}_A\bra{0}_A+\ket{1^\prime}_A\bra{1}_A\bra{1}_A,\\
\hat{O}_B^\prime=&\ket{0^\prime}_B\bra{0}_B\bra{0}_B+\ket{1^\prime}_B\bra{1}_B\bra{1}_B.
\end{split}
\end{equation} 
After error rejection, the state of odd-parity is
\begin{equation*}
\begin{split}
|\rm{odd}\rangle=\frac{\hat{O}_A\hat{O}_B|\Psi\rangle\otimes |\Psi\rangle}{\sqrt{\bra{\Psi}\otimes \bra{\Psi}\hat{O}_B^\dagger \hat{O}_A^\dagger \hat{O}_A\hat{O}_B|\Psi\rangle\otimes |\Psi\rangle}}=\frac{e^{i\phi}}{\sqrt{2}}(\ket{00}+\ket{11}),
\end{split}
\end{equation*} 
and the state of even-parity is
\begin{equation*}
\begin{split}
|\rm{even}\rangle=\frac{\hat{O}_A^\prime \hat{O}_B^\prime|\Psi\rangle\otimes |\Psi\rangle}{\sqrt{\bra{\Psi}\otimes \bra{\Psi}\hat{O}_B^{\prime\dagger} \hat{O}_A^{\prime\dagger} \hat{O}_A^\prime \hat{O}_B^\prime |\Psi\rangle\otimes |\Psi\rangle}}=\frac{1}{\sqrt{2}}(\ket{0^\prime 0^\prime}+e^{2i\phi}\ket{1^\prime 1^\prime}).
\end{split}
\end{equation*} 
This shows that the phase-flip error rate after parity check is dependent on the parity value of the group: $0$ for odd parity and $2\sin^2\phi$ for even parity~\cite{xu2019general}. In general, the iteration formula does not have to hold for our OPER method, which shall only use survived bits from odd-parity group. Luckily, we have shown by de Finetti theorem that after error rejection, the phase-flip error rate of survived untagged bits from odd-parity groups can never be larger than those from even-parity groups, in the limit of infinite number of raw pairs initially~\cite{xu2019general}. However, in a real protocol, we never have infinite number of raw pairs. Since we have used de Finetti theorem, there is no straightforward calculation to take the finite key effects efficiently. If we choose to directly apply the collective-attack model, it will cost a lot of bits in lifting the collective-attack security to coherent-attack security~\cite{christandl2009postselection}. For example, we need to use a failure probability value of the statistical fluctuation as small as $10^{-100}$. Given this request, the final key rate is actually smaller than that of the existing SNS protocol, and, the final key rate can even be 0 if the data size is not unreasonably small. Here we shall give a more efficient way for the issue and present advantageous results with the method with normal, reasonable data size. To do so, we shall first use the Zigzag approach to faithfully bound the phase-flip error rate for those survived untagged bits after OPER. Then we can efficiently take all finite key effects in calculating the finial key.

This paper is arranged as follows. In Sec.~\ref{phaseflip} we present our Zigzag approach to estimate the phase-flip iterative with explicit formulas. In Sec.~\ref{protocol}, we review the SNS protocol and show how to apply the formulas got in Sec.~\ref{phaseflip} to the SNS protocol. In Sec. \ref{asyy}, we show how to apply the OPER method to the asymmetric SNS protocol. In Sec.~\ref{simulation}, we present our numerical simulation results and compare key rates of different methods. The article ends with some concluding remarks. The details of some calculations are shown in the appendix.

\section{Zigzag approach to phase-flip error rate after OPER}\label{phaseflip}
\subsection{Mathematical toolbox and main idea.}\label{mainidea}
For clarity, we first consider the virtual protocol where Alice and Bob share raw entanglement pairs. There is no bit-flip error of all these pairs. After phase-flip error test, they share $2n+k$ of raw entanglement pairs. They randomly choose $2n$ pairs for final key distillation and the left $k$ pairs would be neglected. Denote the state of these $2n$ pairs by $\rho^{2n}$. According to the exponential de Finetti’s representation theorem~\cite{renner2005security,renner2007symmetry}, mathematically, there exists an associate state $\tilde \rho^{2n}$ of state $\rho^{2n}$ which satisfies the following conditions:

\noindent i) The trace distance between $\rho^{2n}$ and $\tilde \rho^{2n}$ is bounded by
\begin{equation}\label{definit}
\Arrowvert \rho^{2n}-\tilde{\rho}^{2n}\Arrowvert\le \varepsilon(r,k)=3k^d e^{-rk/(2n+k)},
\end{equation}

\noindent and ii), state $\tilde \rho^{2n}$ is in the following form:
\begin{equation}\label{definit2}
\tilde{\rho}^{2n}=\int P_{\sigma} \rho_{\sigma}^{2n}d\sigma=\int P_{\sigma} \sigma^{\otimes 2n-r}\otimes \tilde{\rho}_\sigma^r  d\sigma.
\end{equation}
This is to say, it is a state of probabilistic mixture of approximate independent and identically distributed (i.i.d.) states.
The state $\rho_{\sigma}^{2n}$ is called
$ \left( \begin{smallmatrix}
             2n \\ 2n-r
          \end{smallmatrix}  \right) $-i.i.d. state, as only $2n-r$ subsystems of $\rho_{\sigma}^{2n}$ are i.i.d. and the state $\tilde{\rho}_\sigma^r$ is an arbitrary state on $r$ subsystems. For the case we consider here, there is no bit-flip error of raw entanglement pairs, thus the dimension of each subsystem is $d=2$. $\varepsilon(r,k)$ decreases exponentially fast in $rk$. For the practical case, we can set $\varepsilon(r,k)=10^{-13}$ or smaller.

The probability distribution of the number of phase error $m$, of state $\rho^{2n}$, could be observed if a positive operator valued measurement (POVM) $\{\hat{E}_m\}$ was performed~\cite{nielsen2002quantum}. Actually, $p_m$ is just a pure mathematical quantity defined by
\begin{equation}
\{p_m=\tr(\rho^{2n}\hat{E}_m)\}
 \end{equation}
 where $\hat{E}_m$ is the projection operator projecting any state of the $2n$ pairs to the subspace $\mathcal M$ where the number of phase-flip errors for any state takes value $m$ deterministically. 
Mathematically, we also have the probability distribution $q_m$ for the number of phase error $m$, of state $\tilde{\rho}^{2n}$, which is $\{q_m=\tr(\tilde{\rho}^{2n}\hat{E}_m)\}$. According to the property of trace distance, we have
\begin{equation}\label{prob1}
\sum_{m}|p_m-q_m|\le 2\varepsilon(r,k),
\end{equation}
which means the probability distributions of the number of phase error of those two states are almost the same except with a small probability $2\varepsilon(r,k)$. 

To do OPER, they randomly group those $2n$ pairs two by two. For each group of pairs, Alice (Bob) performs the measurement of $\{\hat{O}_A,\hat{O}_A^\prime\}$ ($\{\hat{O}_B,\hat{O}_B^\prime\}$) to her (his) qubits. Since there is no bit-flip error in the raw pairs, only the operators $\hat{O}_A\hat{O}_B$ and $\hat{O}_A^\prime \hat{O}_B^\prime$ would be succeed. And only the results of the operator $\hat{O}_A\hat{O}_B$ would be kept for further key distillation. This completes the OPER. Those pairs after OPER are named as {\em survived pairs}. One can measure the number of phase-flip errors $m_s$ of those {\em survived pairs}.
Physically, there also exists a POVM $\{\hat{E}_m^\prime\}$ which can be taken directly on the initial $2n$ pairs to present the value $m_s$.  The existence of such a POVM can be easily proven.  Therefore, we can also use the probability distributions of $m_s$ for state $\rho^{2n}$ and $\tilde \rho^{2n}$, respectively. We denote the probability distribution of $m_s$ by  $\{p_{m_s}^\prime=\tr(\rho^{2n}\hat{E}_m^\prime)\}$ and $\{q_{m_s}^\prime=\tr(\tilde{\rho}^{2n}\hat{E}_m^\prime)\}$. Similar to Eq.~\eqref{prob1}, we have
\begin{equation}\label{prob2}
\sum_{m_s}|p_{m_s}^\prime-q_{m_s}^\prime|\le 2\varepsilon(r,k).
\end{equation}
\\{\em Remark:} Values of $\{p_m,q_m,p_{m_s}^\prime,q_{m_s}^\prime \}$ present important mathematical properties of states $\rho^{2n}$ and $\tilde\rho^{2n}$. We shall make use of constraints on these values to finally upper bound the phase-flip error rate $e_1^{\prime ph}$ for survived pairs after OPER.

\noindent \textbf{Definition 1}. We define $\Pr\limits_{\mathcal{D}}(x\ge\bar{x})=\sum_{x\ge\bar{x}}\Pr\limits_{\mathcal{D}}(x)$, where $\{\Pr\limits_{\mathcal{D}}(x)\}$ is a probability distribution labelled by $\mathcal{D}$ \textbf{(}$\{\Pr\limits_{\mathcal{D}}(x)\}$ is actually just a probability distribution of the discrete variable $x$, but too many probability distributions are used in the following proof, thus we define such a symbol to clearly show everything more intuitively and vividly by choosing the labelling symbol $\mathcal{D}$ properly. In choosing labelling symbol $\mathcal{D}$, we shall try to always use a symbol characterize the main properties for the probability distribution it labels. Of course every time when we use a different labelling $\mathcal{D}$, we shall always note it clearly that the explicit probability distribution it labels.\textbf{)}. For example, $\Pr\limits_{\{p_m\}}(m\ge\bar{m})=\sum_{m\ge\bar{m}}p_m$.


Here is our main idea in bounding the value of $e_{1}^{\prime ph}$, the phase-flip error rate from those survived bits after OPER. Our goal can be reached once we find tight bound on values of $\{p_{m_s}^\prime\}$.  

\noindent 1. According to the phase-flip error test in the very beginning~\cite{tomamichel2012tight,curty2014finite,jiang2019unconditional}, we can find that the probability for $m\ge \bar M$ is less than $\varepsilon_e$, where $m$ is the number of phase-flip errors from those $2n$ pairs in state $\rho^{2n}$. On the other hand, this probability is just $\sum_{m\ge \widebar M} p_m$. Therefore we can constrain values of $\{p_m\}$ by
\begin{equation}\label{probb1}
\Pr\limits_{\{p_m\}}(m\ge\widebar{M}) \le \varepsilon_e.
\end{equation}

\noindent 2. Since the trace distance between $\rho^{2n}$ and its associate state $\tilde{\rho}^{2n}$ is small, with Eqs. (\ref{prob1}) and \eqref{probb1}, we can further restrict values of $\{q_m\}$ by
\begin{equation}\label{prob3}
\Pr\limits_{\{q_m\}}(m\ge \widebar{M})\le \varepsilon_e+2\varepsilon(r,k).
\end{equation}
 
\noindent 3. With this, we can constrain values of $\{q_{m_s}^\prime\}$ based on the fact that $\tilde{\rho}^{2n}$ is the probabilistic mixture of approximate i.i.d. state. Suppose we have
\begin{equation}\label{prob4}
\Pr\limits_{\{q_{m_s}^\prime\}}(m_s\ge \widebar{M}_s)\le \tilde{\varepsilon}_s,
\end{equation}
where $\widebar{M}_s$ is an estimated value and $\tilde{\varepsilon}_s$ is the corresponding failure probability.

\noindent 4. Again, since the trace distance between $\rho^{2n}$ and its associate state $\tilde{\rho}^{2n}$ is small, with the bounded result in step 3, we can now restrict values of $\{p_{m_s}^\prime\}$ by Eqs.~\eqref{prob2} and ~\eqref{prob4}, which is
\begin{equation}\label{keyeq1}
\Pr\limits_{\{p_{m_s}^\prime\}}(m_s\ge \widebar{M}_s)\le \tilde{\varepsilon}_s+2\varepsilon(r,k)= \varepsilon_s.
\end{equation}

\noindent 5. Finally, with constraint on $\{p_{m_s}^\prime\}$ in step 4, we can get the upper bound of the phase-flip error rate $e_{1}^{\prime ph}$ with a failure probability $\varepsilon_s$, which is
\begin{equation}
e_{1}^{\prime ph}=\frac{\widebar{M}_s}{n_1^\prime},
\end{equation} 
where $n_{1}^\prime$ is the number of survived pairs after OPER.

Among all values and failure probability appeared in Eqs.~(\ref{prob1}-\ref{prob4}), $\widebar{M}$ and $\varepsilon_e$ could be get from the phase-flip error test in the very beginning~\cite{tomamichel2012tight,curty2014finite,jiang2019unconditional}, and $\varepsilon(r,k)$ is set as we wanted. The only unknown values are $\widebar{M}_s$ and $\tilde{\varepsilon}_s$. Our task is now reduced to find the explicit values of $\widebar{M}_s$ and its corresponding failure probability $\tilde{\varepsilon}_s$.

\subsection{The values of $\bar{M}_s$ and its corresponding failure probability $\tilde{\varepsilon}_s$}

\begin{lemma} \label{lemma1}
Let $\{\Pr\limits_{c_1}(x_1)\}$ be an $a_1$-fold Bernoulli distribution where the success probability of obtaining outcome $'1'$ of every Bernoulli trial is $c_1$. Let $\{\Pr\limits_{\mathcal{D}_2}(x_2)\}$ be an arbitrary probability distribution on $a_2$ random variables with outcomes $'0'$ or $'1'$ and $a_2\le a_1$. Let $x$ be a new discrete variable which satisfies $x=x_1+x_2$, and we denote the probability distribution of $x$ as $\{\Pr\limits_{\mathcal{D}}(x)\}$. Then we have
\begin{equation}\label{keyle1}
\begin{split}
&\Pr\limits_{\mathcal{D}}(x\ge\bar{x})\ge \Pr\limits_{c_1}(x_1\ge\bar{x}),\\
&\Pr\limits_{\mathcal{D}}(x\ge\bar{x})\le \Pr\limits_{c_1}(x_1\ge\bar{x}-a_2),
\end{split}
\end{equation}
where $\bar{x}$ is a specific value in $[a_2,a_1]$.
\end{lemma}

\begin{proof}
According to the definition of $\Pr\limits_{\mathcal{D}}(x\ge\bar{x})$, we have
\begin{equation}\label{eqpp1}
\Pr\limits_{\mathcal{D}}(x\ge\bar{x})=\sum_{x_2=0}^{a_2}\Pr\limits_{\mathcal{D}_2}(x_2)\sum_{x_1=\bar{x}-x_2}^{a_1}\Pr\limits_{c_1}(x_1).
\end{equation}
And obviously, we have
\begin{equation}\label{eqpp2}
\begin{split}
&\sum_{x_1=\bar{x}-x_2}^{a_1}\Pr\limits_{c_1}(x_1)\ge\sum_{x_1=\bar{x}}^{a_1}\Pr\limits_{c_1}(x_1),\\
&\sum_{x_1=\bar{x}-x_2}^{a_1}\Pr\limits_{c_1}(x_1)\le\sum_{x_1=\bar{x}-a_2}^{a_1}\Pr\limits_{c_1}(x_1).
\end{split}
\end{equation}
Combine Eqs.~(\ref{eqpp1},\ref{eqpp2}) and the fact that $\sum_{x_2=0}^{a_2}\Pr\limits_{\mathcal{D}_2}(x_2)=1$, we can get Eq.~\eqref{keyle1}. This ends the proof of Lemma \ref{lemma1}.
\end{proof}

Lemma \ref{lemma1} shows that if $a_1\gg a_2$, the statistical property of $x$ is almost determined by its i.i.d. part $x_1$. This also is the original intention of exponential de Finetti’s representation theorem~\cite{renner2005security,renner2007symmetry}.

\begin{lemma} \label{lemma2}
Let $\{\Pr\limits_{c_1}(x_1)\}$ be an $a_1$-fold Bernoulli distribution where the success probability of obtaining outcome $'1'$ of every Bernoulli trial is $c_1$. Let $\{\Pr\limits_{c_2}(x_1)\}$ be an $a_1$-fold Bernoulli distribution where the success probability of obtaining outcome $'1'$ of every Bernoulli trial is $c_2$. If $c_1\ge c_2$, we have
\begin{equation}\label{keyle2}
\Pr\limits_{c_1}(x_1\ge\bar{x})\ge \Pr\limits_{c_2}(x_1\ge\bar{x})
\end{equation}
where $\bar{x}$ is a specific value in $[0,a_1]$.
\end{lemma}
\begin{proof}
Lemma \ref{lemma2} is a directed conclusion of binomial distribution.
\end{proof}

The state $\sigma$ of a subsystem in $\rho_\sigma^{2n}$ is a $2$-dimensional state, which can be express as
\begin{equation}
\begin{split}
\sigma=&\cos^2\theta\oprod{00}{00}+\sin^2\theta\oprod{11}{11}+\alpha e^{-i\beta}\oprod{11}{00}+\alpha e^{i\beta}\oprod{00}{11},
\end{split}
\end{equation}
where $0$ and $1$ mean qubit $0$ and $1$, and the first qubit belongs to Alice and the second qubit belongs to Bob. $\theta,\alpha,\beta$ are three arbitrary real numbers and satisfy
\begin{equation}
\alpha^2\le \sin^2\theta\cos^2\theta.
\end{equation}

\begin{lemma} \label{lemma3}
Let $\mean{e_\sigma}$ be the probability that an error occurs if Alice and Bob measures their qubit of state $\sigma$ in $X$ basis. Let $\mean{E_\sigma}$ be the probability that an error occurs if Alice and Bob measures their qubit in $X$ basis after OPER performed on $\sigma^{\otimes 2}$ (the detail of the definition of $\mean{E_\sigma}$ is shown in Eq.~\eqref{Esigma}). Let $\bar{e}$ be a specific value in $[0,0.5]$. If $\mean{e_\sigma}\le \bar{e}$, we have
\begin{equation}
\mean{E_\sigma}\le \bar{e}(1-\bar{e}).
\end{equation}
\end{lemma}

\begin{proof}
We have proof Lemma \ref{lemma3} in Ref.\cite{xu2019general}. For completeness, we write the proof again.

According to the definition of $\mean{e_\sigma}$, we have
\begin{equation}\label{esigma}
\mean{e_\sigma}=\tr(\hat{M}_1\sigma \hat{M}_1^\dagger+\hat{M}_2\sigma \hat{M}_2^\dagger)=\frac{1}{2}(1-2\alpha\cos \beta),
\end{equation}
where
\begin{equation}
\begin{split}
\hat{M}_1&=\ket{+}_A\bra{+}_A\otimes \ket{-}_B\bra{-}_B,\\
\hat{M}_2&=\ket{-}_A\bra{-}_A\otimes \ket{+}_B\bra{+}_B,
\end{split}
\end{equation}
and
\begin{equation}
\begin{split}
\ket{+}&=\frac{1}{\sqrt{2}}(\ket{0}+\ket{1}),\\
\ket{-}&=\frac{1}{\sqrt{2}}(\ket{0}-\ket{1}).
\end{split}
\end{equation}

With the operator defined in Eq.~\eqref{opero}, we have
\begin{equation}\label{Esigma}
\begin{split}
\mean{E_\sigma}=&\tr(\hat{M}_1\hat{O}_A\hat{O}_B\sigma^{\otimes 2} \hat{O}_B^\dagger \hat{O}_A^\dagger \hat{M}_1^\dagger+\hat{M}_2\hat{O}_A\hat{O}_B\sigma^{\otimes 2} \hat{O}_B^\dagger \hat{O}_A^\dagger \hat{M}_2^\dagger)\\
=&\sin^2\theta\cos^2\theta-\alpha^2.
\end{split}
\end{equation}

Directly, with Eqs.~\eqref{esigma} and \eqref{Esigma}, we have
\begin{equation}
\mean{E_\sigma}\le \mean{e_\sigma}(1-\mean{e_\sigma})\le \bar{e}(1-\bar{e}).
\end{equation}
This ends the proof of Lemma~\ref{lemma3}.
\end{proof}

The state $\rho_\sigma^{2n}$ is composed by two independent parts: one part is the i.i.d. state $\sigma^{\otimes 2n-r}$ and the other part is an arbitrary state $\tilde{\rho}_\sigma^{r}$. If Alice and Bob measure the phase error of $\rho_\sigma^{2n}$, and the outcome of the number of phase error is $m$, then $m$ is also composed by two parts $m=m_1+m_2$ where $m_1$ is the number of phase error of state $\sigma^{\otimes 2n-r}$ and  $m_2$ is the number of phase error of state $\tilde{\rho}_\sigma^{r}$. We denote the probability distribution of the number of phase error of state $\rho_\sigma^{2n}$ as $\{\Pr\limits_{\mathcal{D}_{\sigma r}}(m)\}$ and the probability distribution of the number of phase error of state $\tilde{\rho}_\sigma^{r}$ as $\{\Pr\limits_{\mathcal{D}_{r}}(m_2)\}$. Obviously, the probability distribution of the number of phase error of state $\sigma^{\otimes 2n-r}$ is a $(2n-r)$-fold Bernoulli distribution where the success probability of obtaining outcome $'1'$ of every Bernoulli trial is $\mean{e_{\sigma}}$, and we denote this probability distribution as $\{\Pr\limits_{\mean{e_{\sigma}}}(m_1)\}$.

\begin{theorem}\label{theorem1}
Let $\Pr\limits_{\mean{e_{\tau}}}(m_1\ge\widebar{M})=\xi_\tau$ where $\{\Pr\limits_{\mean{e_{\tau}}}(m_1)\}$ is a $(2n-r)$-fold Bernoulli distribution where the success probability of obtaining outcome $'1'$ of every Bernoulli trial is $\mean{e_{\tau}}$, we have
\begin{equation}
P_1=\int_{\sigma\in\{\sigma:\mean{e_{\sigma}}\ge \mean{e_{\tau}}\}} P_{\sigma}d\sigma\le \frac{\varepsilon_e+2\varepsilon(r,k)}{\xi_\tau}.
\end{equation}
\end{theorem}

\begin{proof}
\begin{equation*}
\begin{split}
\Pr\limits_{\{q_m\}}(m\ge \widebar{M})=&\int P_{\sigma}\Pr\limits_{\mathcal{D}_{\sigma r}}(m\ge\widebar{M})d\sigma\\
\ge &\int_{\sigma\in\{\sigma:\mean{e_{\sigma}}\ge \mean{e_{\tau}}\}} P_{\sigma}\Pr\limits_{\mathcal{D}_{\sigma r}}(m\ge\widebar{M})d\sigma\\
\ge &\int_{\sigma\in\{\sigma:\mean{e_{\sigma}}\ge \mean{e_{\tau}}\}} P_{\sigma}\Pr\limits_{\mean{e_{\sigma}}}(m_1\ge\widebar{M})d\sigma\\
\ge &\int_{\sigma\in\{\sigma:\mean{e_{\sigma}}\ge \mean{e_{\tau}}\}} P_{\sigma}\Pr\limits_{\mean{e_{\tau}}}(m_1\ge\widebar{M})d\sigma\\
=&P_1\xi_\tau,
\end{split}
\end{equation*}
where we apply Lemma \ref{lemma1} for the second inequality and apply Lemma \ref{lemma2} for the third inequality. Combine with Eq.~\eqref{prob3}, we have $P_1\xi_\tau\le \varepsilon_e+2\varepsilon(r,k)$. This ends the proof of Theorem~\ref{theorem1}.
\end{proof}

For the state $\rho_\sigma^{2n}$, the OPER process be regarded as independently happening in two individual states: one is the state $\sigma^{\otimes 2n-2r}$ and the other is the state $\sigma^{\otimes r}\otimes\tilde{\rho}_\sigma^{r}$. If Alice and Bob measure the phase error of $\rho_\sigma^{2n}$ after OPER, and the outcome of the number of phase error is $m_s$, then $m_s$ is composed by two parts $m_s=m_{s1}+m_{s2}$ where $m_{s1}$ is the number of phase error of state $\sigma^{\otimes 2n-2r}$ after OPER and $m_{s2}$ is the number of phase error of state $\sigma^{\otimes r}\otimes\tilde{\rho}_\sigma^{r}$ after OPER. We denote the probability distribution of the number of phase error of state $\rho_\sigma^{2n}$ after OPER as $\{\Pr\limits_{\mathcal{D}_{\sigma r}^\prime}(m_s)\}$ and the probability distribution of the number of phase error of state $\sigma^{\otimes r}\otimes\tilde{\rho}_\sigma^{r}$ after OPER as $\{\Pr\limits_{\mathcal{D}_{r}^\prime}(m_{s2})\}$. Obviously, the probability distribution of the number of phase error of state $\sigma^{\otimes 2n-2r}$ after OPER is a $(n-r)$-fold Bernoulli distribution where the success probability of obtaining outcome $'1'$ of every Bernoulli trial is $\mean{E_{\sigma}}$, and we denote this probability distribution as $\{\Pr\limits_{\mean{E_{\sigma}}}(m_{s1})\}$.

\begin{theorem}\label{theorem2}
Let $\mean{E_{\tau}}=\mean{e_{\tau}}(1-\mean{e_{\tau}})$, and $\{\Pr\limits_{\mean{E_{\tau}}}(m_{s1})\}$ be a $(n-r)$-fold Bernoulli distribution where the success probability of obtaining outcome $'1'$ of every Bernoulli trial is $\mean{E_{\tau}}$. And let $\tilde{\xi}_\tau$ and $\widebar{M}_s$ satisfy $\Pr\limits_{\mean{E_{\tau}}}(m_{s1}\ge\widebar{M}_s-r)=\tilde{\xi}_\tau$. We have
\begin{equation}
\Pr\limits_{\{q_{m_s}^\prime\}}(m_s\ge \widebar{M}_s)\le \tilde{\varepsilon}_s=\tilde{\xi}_\tau+[1-\tilde{\xi}_\tau]\frac{\varepsilon_e+2\varepsilon(r,k)}{\xi_\tau}.
\end{equation}
\end{theorem}
\begin{proof}
\begin{equation*}
\begin{split}
&\Pr\limits_{\{q_{m_s}^\prime\}}(m_s\ge \widebar{M}_s)=\int P_{\sigma}\Pr\limits_{\mathcal{D}_{\sigma r}^\prime}(m_s\ge\widebar{M}_s)d\sigma\\
&\le \int P_{\sigma}\Pr\limits_{\mean{E_{\sigma}}}(m_{s1}\ge\widebar{M}_s-r)d \sigma \\
&\le \int_{\sigma\in\{\sigma:\mean{e_{\sigma}}\le \mean{e_{\tau}}\}} P_{\sigma}\Pr\limits_{\mean{E_{\sigma}}}(m_{s1}\ge\widebar{M}_s-r)d\sigma+P_1\\
&\le \int_{\sigma\in\{\sigma:\mean{e_{\sigma}}\le \mean{e_{\tau}}\}} P_{\sigma}\Pr\limits_{\mean{E_{\tau}}}(m_{s1}\ge\widebar{M}_s-r)d\sigma+P_1\\
&= (1-P_1)\tilde{\xi}_\tau+P_1,
\end{split}
\end{equation*}
where we apply Lemma \ref{lemma1} for the first inequality, apply the fact that all probabilities are less than or equal to $1$ for the second inequality, and apply Lemma \ref{lemma2} and Lemma \ref{lemma3} for the third inequality. Combine with Theorem \ref{theorem1}, we have
\begin{equation}
\Pr\limits_{\{q_{m_s}^\prime\}}(m_s\ge \widebar{M}_s)\le\tilde{\xi}_\tau+[1-\tilde{\xi}_\tau]\frac{\varepsilon_e+2\varepsilon(r,k)}{\xi_\tau}.
\end{equation}
This ends the proof of Theorem \ref{theorem2}.
\end{proof}

With Theorem \ref{theorem1} and Theorem \ref{theorem2}, we can calculate a group of $(\widebar{M}_s,\varepsilon_s)$ that satisfies Eq.~\eqref{keyeq1}. For example, if we have known a group of $(\widebar{M},\varepsilon_e)$ that satisfies Eq.~\eqref{probb1}, we can set $\xi_\tau=10^{-2},\tilde{\xi}_\tau=10^{-10}$. By solving the equation $\Pr\limits_{\mean{e_{\tau}}}(m_1\ge\widebar{M})=10^{-2}$, we can get the value of $\mean{e_{\tau}}$. With the definition $\mean{E_{\tau}}=\mean{e_{\tau}}(1-\mean{e_{\tau}})$, by solving the equation $\Pr\limits_{\mean{E_{\tau}}}(m_{s1}\ge\widebar{M}_s-r)=10^{-10}$, we can get the value of $\widebar{M}_s$. Approximately, we have $\varepsilon_s=10^{-10}+10^2[\varepsilon_e+2\varepsilon(r,k)]+2\varepsilon(r,k)$ and
\begin{equation}\label{eqms}
\begin{split}
&\widebar{M}_s=(n-r)\mean{e_{\tau}}(1-\mean{e_{\tau}})+6.36\sqrt{(n-r)\mean{e_{\tau}}(1-\mean{e_{\tau}})}+r,\\
&\mean{e_{\tau}}=\frac{\widebar{M}-2.33\sqrt{\widebar{M}}}{2n-r}.
\end{split}
\end{equation}

\section{The finite key effect of SNS protocol with OPER}\label{protocol}
\subsection{The SNS protocol and its parameter estimation}
We consider the 4-intensity SNS protocol~\cite{yu2019sending,jiang2019unconditional}. In this protocol, Alice and Bob send $N$ pulse pairs to Charlie and get a series of data. In each time window, Alice (Bob) randomly chooses the decoy window or signal window with probabilities $1-p_{z}$ and $p_{z}$ respectively. If the decoy window is chosen, Alice (Bob) randomly prepares the pulse of vacuum state or WCS state $\ket{e^{i\theta_A}\sqrt{\mu_{1}}}$ or $\ket{e^{i\theta_A^\prime}\sqrt{\mu_{2}}}$ (vacuum state, or WCS state $\ket{e^{i\theta_B}\sqrt{\mu_{1}}}$ or $\ket{e^{i\theta_B^\prime}\sqrt{\mu_{2}}}$) with probabilities $p_{0}$, $p_{1}$ and $1-p_{0}-p_{1}$, respectively, where $\theta_A,\theta_A^\prime,\theta_B$ and $\theta_B^\prime$ are different in different windows, and are random in $[0,2\pi)$. If the signal window is chosen, Alice (Bob) randomly chooses bit $1$ or $0$ ($0$ or $1$) with probabilities $\epsilon$ and $1-\epsilon$, respectively. If bit $1$ ($0$) is chosen, Alice (Bob) prepares a phase-randomized WCS pulse with intensity $\mu_{z}$. If bit $0$ ($1$) is chosen, Alice (Bob) prepares a vacuum pulse. This is said to be the sending or not-sending. For Alice (Bob), bit value 1 (0) is due to her (his) decision on sending out a phase randomized coherent state with intensity $\mu_z$ for Charlie, while bit value 0 (1) is corresponding to her (his) decision of not-sending, i.e., sending out a vacuum.

Then Alice and Bob send their prepared pulses to Charlie. Charlie is assumed to perform interferometric measurements on the received pulses and announces the measurement results to Alice and Bob. If one and only one detector clicks in the measurement process, Charlie also tells Alice and Bob which detector clicks, and Alice and Bob take it as an one-detector heralded event. Alice and Bob repeat the above process for $N$ times and collect all the data with one-detector heralded events and discard all the others.

The next process is the parameter estimation. To clearly show how this process is carried out, we have the following definitions.

\noindent\textbf{Definition 2}. If both Alice and Bob choose the signal window, it is a $Z$ window. If both Alice and Bob choose the decoy window, and the states Alice and Bob prepared are in the same intensity and their phases satisfy the post-selection criterion~\cite{hu2019general}, it is an $X$ window. The one-detector heralded events of $X$ windows and $Z$ windows are called effective events. And Alice and Bob respectively get $n_t-$bit strings, $Z_A$ and $Z_B$, formed by the corresponding bits of effective events of $Z$ windows.

\noindent\textbf{Definition 3}. For an effective event in the $Z$ windows, if it is caused by the event that only one party of Alice and Bob decides sending out a phase-randomized WCS pulse and he (she) actually sends out a single photon state from the view point of decoy state method, it is an untagged event. Its corresponding bit is an untagged bit.

\textbf{Parameter estimation~\cite{yu2019sending,hu2019general,jiang2019unconditional}.} The data of all the one-detector heralded events except the effective events of $Z$ windows in this protocol are used to estimate the expected value of the lower bounds of the number of untagged $0$-bits, $\mean{\underline{n_{01}}}$, and untagged $1$-bits, $\mean{\underline{n_{10}}}$. The details of how to estimate the lower bound of $\mean{\underline{n_{01}}}$ and $\mean{\underline{n_{10}}}$ are shown in Appendix A. We denote the expected value of the lower bounds of the number of untagged bits as $\mean{\underline{n_{1}}}=\mean{\underline{n_{01}}}+\mean{\underline{n_{10}}}$. 

Also the one-detector heralded events could be used to estimate the expected value of the upper bound of phase-flip error rate, $\mean{\overline{e_1^{ph}}}$, of the untagged events. There are two methods to estimate $\mean{\overline{e_1^{ph}}}$. The first method is shown in Eq.~\eqref{e1} of Appendix A, where we apply the Chernoff bound for two times to estimate the upper bound of the first term and the lower bound of the second term of the numerator, and this is also the methods we use in our prior articles~\cite{yu2019sending,jiang2019unconditional}. The second method is shown in Appendix B, where we apply the improved version of McDiarmid inequality~\cite{chau2019application} to the numerator of  Eq.~\eqref{e1}. The McDiarmid inequality allows us to treat the numerator as a whole, and directly get the upper bound of the numerator in a fixed failure probability. This method effectively reduces the effect of statistical fluctuation of estimating $\mean{\overline{e_1^{ph}}}$ and improves the key rate especially when the key size is small. The details of the second method are shown in Appendix B.

\subsection{The data post-processing of SNS protocol}\label{datapost}
The data post-processing of SNS protocol contains two independent processes: the error correction and the privacy amplification. The goal of error correction is to correct all the different bits in $Z_A$ and $Z_B$. And the privacy amplification is to distill a shorter but more secure final key from the raw key according to the formula of key rate.

If we directly correct the different bits in $Z_A$ and $Z_B$, a large number of raw keys would be cost and decrease the key rate. The OPER could be used to improve the key rate. And further, we can use the active OPER, which is also called active odd-parity paring (AOPP) in Ref.~\cite{xu2019general} to further improve the key rate. And the security of AOPP is equivalent to the security of OPER~\cite{xu2019general}. In AOPP, Bob first actively pairs the bits $0$ with bits $1$ of the raw key string $Z_B$, and get $2n_g$ pairs, and then randomly splits those pairs into two equal parts. For each part, Alice computes the parities of those $n_g$ pairs and announces them to Bob, then Alice and Bob keep the pairs with parity $1$ and discard the pairs with parity $0$. Finally, Alice and Bob randomly keep one bits from those survived pairs and form two new $n_t^\prime$-bits strings, which would be used to perform error correction and privacy amplification. As shown in Ref.~\cite{xu2019general}, there is no mutual information between the split two parts and each of those two parts alone could be regarded as the result of a virtual OPER process performed on $un_t$ bits of $Z_A$ and $Z_B$, where
\begin{equation}
u=\frac{n_g}{n_{odd}},
\end{equation}
where ${n_{odd}}$ is the number of pairs with odd parity if Bob randomly groups all the bits in $Z_B$ two by two, and both $n_g$ and $n_{odd}$ are observed values in practice. Thus we could first get the formula of phase-flip error rate of the survived untagged bits after OPER, and then apply this formula to each part of AOPP. 


In Sec.~\ref{phaseflip}, we have shown how to get the phase-flip error rate of the survived untagged bits after OPER, and here we would show how to get the values of $n,k,r$ and so on in the instance of SNS protocol. The value of $n$ is the number of untagged bit pairs after random pairing, where the untagged bit pairs are the pairs formed by two untagged bits. And we have
\begin{equation}\label{esn}
n=\varphi^L(\frac{\mean{\underline{n_{1}}}}{n_t}\frac{\mean{\underline{n_{1}}}}{n_t}\frac{un_t}{2}),
\end{equation}
where $\varphi^L(x)$ is defined in Eq.~\eqref{observL}. The value of $k$ is the number of neglected untagged bits, thus we can take the untagged bits that are paired with tagged bits as the neglected bits, thus we have
\begin{equation}\label{esk}
k=\varphi^L(u\mean{\underline{n_{1}}}-\frac{\mean{\underline{n_{1}}}}{n_t}\frac{\mean{\underline{n_{1}}}}{n_t}un_t).
\end{equation} 
With Eq.\eqref{e1}, we can get the value of $\mean{\overline{e_1^{ph}}}$. Particularly, the failure probability of Chernoff bound used in Eq.\eqref{e1} is set as $10^{-13}$. And we have
\begin{equation}
\widebar{M}=\varphi^U(2n\mean{\overline{e_1^{ph}}}),
\end{equation}
where $\varphi^U(x)$ is defined in Eq.~\eqref{observ} and we also set the failure probability as $10^{-13}$, thus $\varepsilon_e=3\times 10^{-13}$. All the failure probabilities of Chernoff bound used in other equations are set as $10^{-10}$. We set $\varepsilon(r,k)=10^{-13}$ and we have
\begin{equation}
r=\frac{2n+k}{k}\ln\frac{3k^2}{\varepsilon(r,k)}.
\end{equation} 
We set $\xi_\tau=10^{-2},\tilde{\xi}_\tau=10^{-10}$, and then we can calculate the value of $\widebar{M}_s$ with Eq.\eqref{eqms}. And we have $\varepsilon_{s}=1.5\times 10^{-10}$. Finally, we get the upper bound of the phase-flip error rate of the survived untagged bits after OPER with a failure probability $\varepsilon_{s}=1.5\times 10^{-10}$, which is 
\begin{equation}
e_{1}^{\prime ph}=\frac{\widebar{M}_s}{n_1^\prime},
\end{equation}  
where $n_1^\prime$ is the number of the survived untagged bits after OPER, and 
\begin{equation}
n_1^\prime=\varphi^L(\frac{\mean{\underline{n_{01}}}}{n_t}\frac{\mean{\underline{n_{10}}}}{n_t}un_t).
\end{equation}

With all those values, we now can calculate the key rate $R$ of AOPP with the formulas in Ref.~\cite{jiang2019unconditional}, which is 
\begin{equation}\label{r2}
\begin{split}
R=&\frac{2}{N}\{n_1^\prime[1-h(e_{1}^{\prime ph})]-fn_t^\prime h(E^\prime)-\log_2{\frac{2}{\varepsilon_{cor}}}-2\log_2{\frac{1}{\sqrt{2}\varepsilon_{PA}\hat{\varepsilon}}}\}.
\end{split}
\end{equation}
where $h(x)=-x\log_2x-(1-x)\log_2(1-x)$ is the Shannon entropy, $E^\prime$ is the bit-flip error rate of the remain bits after AOPP, $\varepsilon_{cor}$ is the failure probability of error correction, $\varepsilon_{PA}$ is the failure probability of privacy amplification, and $\hat{\varepsilon}$ is the coefficient while using the chain rules of smooth min- and max- entropy~\cite{vitanov2013chain}. And we set $\varepsilon_{cor}=\varepsilon_{PA}=\hat{\varepsilon}=10^{-10}$.

With the formula of Eq.~\eqref{r2}, the protocol is $2\varepsilon_{tol}$-secure, and $\varepsilon_{tol}=\varepsilon_{cor}+\varepsilon_{sec}$, where $\varepsilon_{sec}=2\hat{\varepsilon}+4\varepsilon_s+\varepsilon_{PA}+\varepsilon_{n_1^\prime}+\varepsilon_{nk}$. Here, $\varepsilon_{n_1^\prime}$ is the probability that the real value of the number of survived untagged bits is smaller than $n_1^\prime$, $\varepsilon_{nk}$ is the failure probability while we estimated the value of $n$ and $k$ in Eqs.~\eqref{esn} and \eqref{esk}. And we have $\varepsilon_{n_1^\prime}=6\times 10^{-10}$ for we use the Chernoff bound for $6$ times to estimate $n_1^\prime$, and with the similar reason, we have $\varepsilon_{nk}=2\times 10^{-10}$. Finally, we have $\varepsilon_{tol}=1.8\times 10^{-9}$.

\section{The asymmetric SNS protocol with OPER and finite key effect}\label{asyy}
In the practical application of SNS protocol, the distance between Alice and Charlie can be different from the distance between Bob and Charlie. To solve this problem, the theory of asymmetric SNS protocol is proposed in Ref.~\cite{hu2019general}. The preparation and measurement step of asymmetric SNS protocol is the same with that of the original SNS protocol, except the source parameters of Alice and Bob in asymmetric SNS protocol are not the same. In this part, we use the subscript $^\prime$ to indicate Bob's source parameters. For example, $p_z$ is the probability that Alice chooses the signal window and $p_z^\prime$ is the probability that Bob chooses the signal window; $\mu_z$ is the intensity of phase-randomized WCS if Alice decides sending in her signal windows and $\mu_z^\prime$ is the intensity of phase-randomized WCS if Bob decides sending in his signal windows.

Note that the original SNS protocol~\cite{wang2018twin} and its improved one~\cite{xu2019general} based on symmetric source parameters for Alice and Bob, i.e., they use the same values for the sending probabilities and light intensities. As was shown in Ref.~\cite{hu2019general}, the SNS protocol is also secure with asymmetric source parameters given the following mathematical constraint:
\begin{equation}
\frac{\mu_{1}}{\mu_{1}^\prime}=\frac{\epsilon(1-\epsilon^\prime)\mu_{z}e^{-\mu_{z}}}{\epsilon^\prime(1-\epsilon)\mu_{z}^\prime e^{-\mu_{z}^\prime}}.
\end{equation}
With this condition, light intensity chosen by Alice and that chosen by Bob can be different.

After Alice and Bob repeat the preparation and measurement steps of asymmetric SNS protocol for $N$ times, they can perform the same error correction and privacy amplification steps as shown in Sec.~\ref{datapost}. And the formula of key rate is the same as Eq.~\eqref{r2}. The major differences between the original SNS protocol and asymmetric SNS protocol are the forms of the formulas of the lower bound of the number of untagged bits and the upper bound of phase-flip error rate, which are shown in the appendix \ref{calculation}.

\section{Numerical simulation}\label{simulation}
In this part, we show the results of numerical simulation of SNS protocol with AOPP, including the symmetric and asymmetric cases, and compare with the results of the original SNS protocol~\cite{jiang2019unconditional,hu2019general}. The results of this work (A) are calculated by the methods shown in Appendix A, and the results of this work (B) are calculated by the similar methods except that we apply the methods shown in Appendix B to estimate the phase-flip error rate. The two bounds of repeater-less key rate used here are called 'PLOB-1' and 'PLOB-2'~\cite{pirandola2017fundamental}, where PLOB-1 is the absolute PLOB bound and 'PLOB-2' is the practical bound. 


\begin{table}
\begin{ruledtabular}
\begin{tabular}{ccccc}
$p_d$& $e_d$ &$\eta_d$ & $f$ & $\alpha_f$ \\
\hline
$1.0\times10^{-8}$& $3\%$  & $30.0\%$ & $1.1$ & $0.2$ \\ 
\end{tabular}
\end{ruledtabular}
\caption{List of experimental parameters used in numerical simulations. Here $p_d$ is the dark count rate of Charlie's detectors; $e_d$ is the misalignment-error probability; $\eta_d$ is the detection efficiency of Charlie's detectors; $f$ is the error correction inefficiency; $\alpha_f$ is the fiber loss coefficient ($dB/km$).}\label{exproperty}
\end{table}

We use the linear model to simulate the observed values of experiment such as $S_{\kappa\zeta}$ which are defined in Appendix \ref{calculation}~\cite{jiang2019unconditional}, with the experimental parameters list in Table.~\ref{exproperty}. And the simulation method of $n_g,n_t^\prime,n_{ood}$ and $E^\prime$ are shown in Appendix \ref{simulationmethod}.  Without loss of generality, we assume the property of Charlie's two detectors are the same. The distance between Alice and Charlie is $L_A$, and the distance between Bob and Charlie is $L_B$. The total distance between Alice and Bob is $L=L_A+L_B$. In our numerical simulation, we set $L_A=L_B$ for the symmetric case and $L_A-L_B=$constant for the asymmetric case.

\begin{figure}
\centering
\includegraphics[width=8cm]{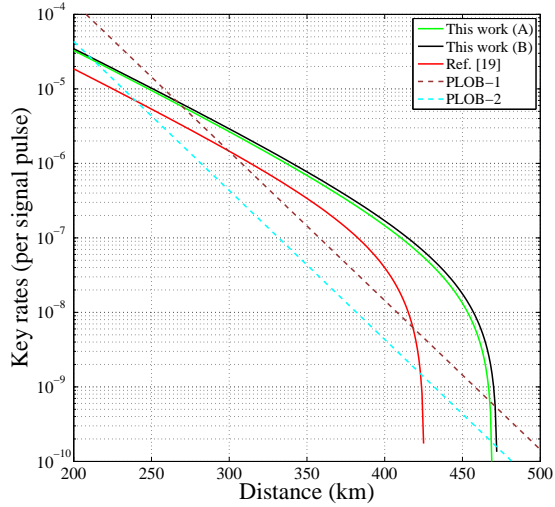}
\caption{The optimal key rates (per pulse) versus transmission distance (the distance between Alice and Bob) with total number of pulses $N=10^{12}$. Here we set $L_A=L_B$ and the source parameters of Alice and Bob are all the same. The experimental parameters that we used in the numerical simulation are listed in Table \ref{exproperty}. The black solid line is the optimized results of this work (B). The green solid line is the optimized results of this work (A). And the red solid line is the optimized results of Ref.~\cite{jiang2019unconditional}.}\label{figure1}
\end{figure}

\begin{figure}
\centering
\includegraphics[width=8cm]{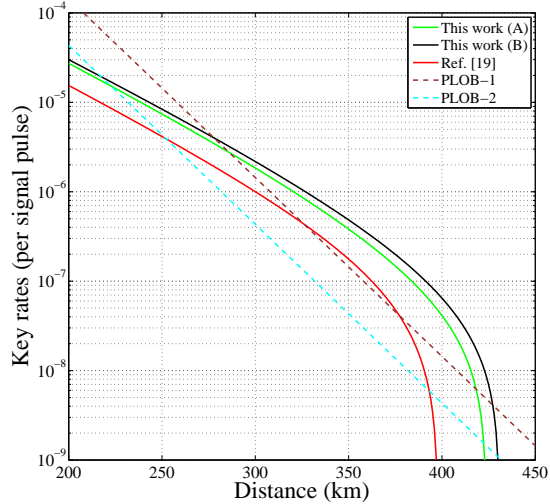}
\caption{The optimal key rates (per pulse) versus transmission distance with total number of pulses $N=10^{11}$. Here we set $L_A=L_B$ and the source parameters of Alice and Bob are all the same. The experimental parameters that we used in the numerical simulation are listed in Table \ref{exproperty}. The black solid line is the optimized results of this work (B). The green solid line is the optimized results of this work (A). And the red solid line is the optimized results of Ref.~\cite{jiang2019unconditional}. }\label{figure2}
\end{figure}

\begin{figure}
\centering
\includegraphics[width=8cm]{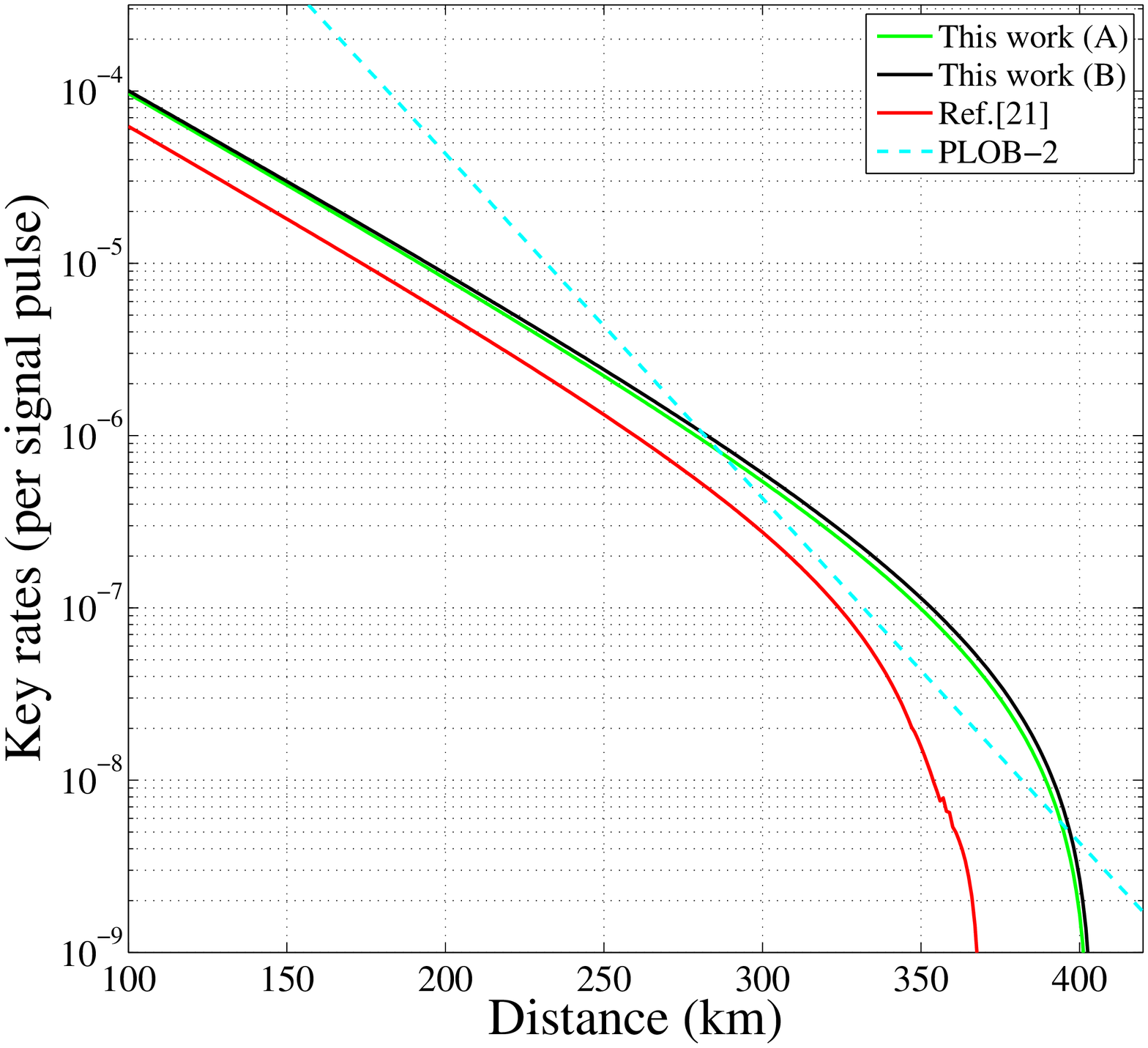}
\caption{The optimal key rates (per pulse) versus transmission distance with total number of pulses $N=10^{12}$. Here we set $L_A-L_B=100$ km. The experimental parameters that we used in the numerical simulation are listed in Table \ref{exproperty}. The black solid line is the optimized results of this work (B). The green solid line is the optimized results of this work (A). And The red solid line is the optimized results of Ref.~\cite{hu2019general}. }\label{figure3}
\end{figure}

\begin{figure}
\centering
\includegraphics[width=8cm]{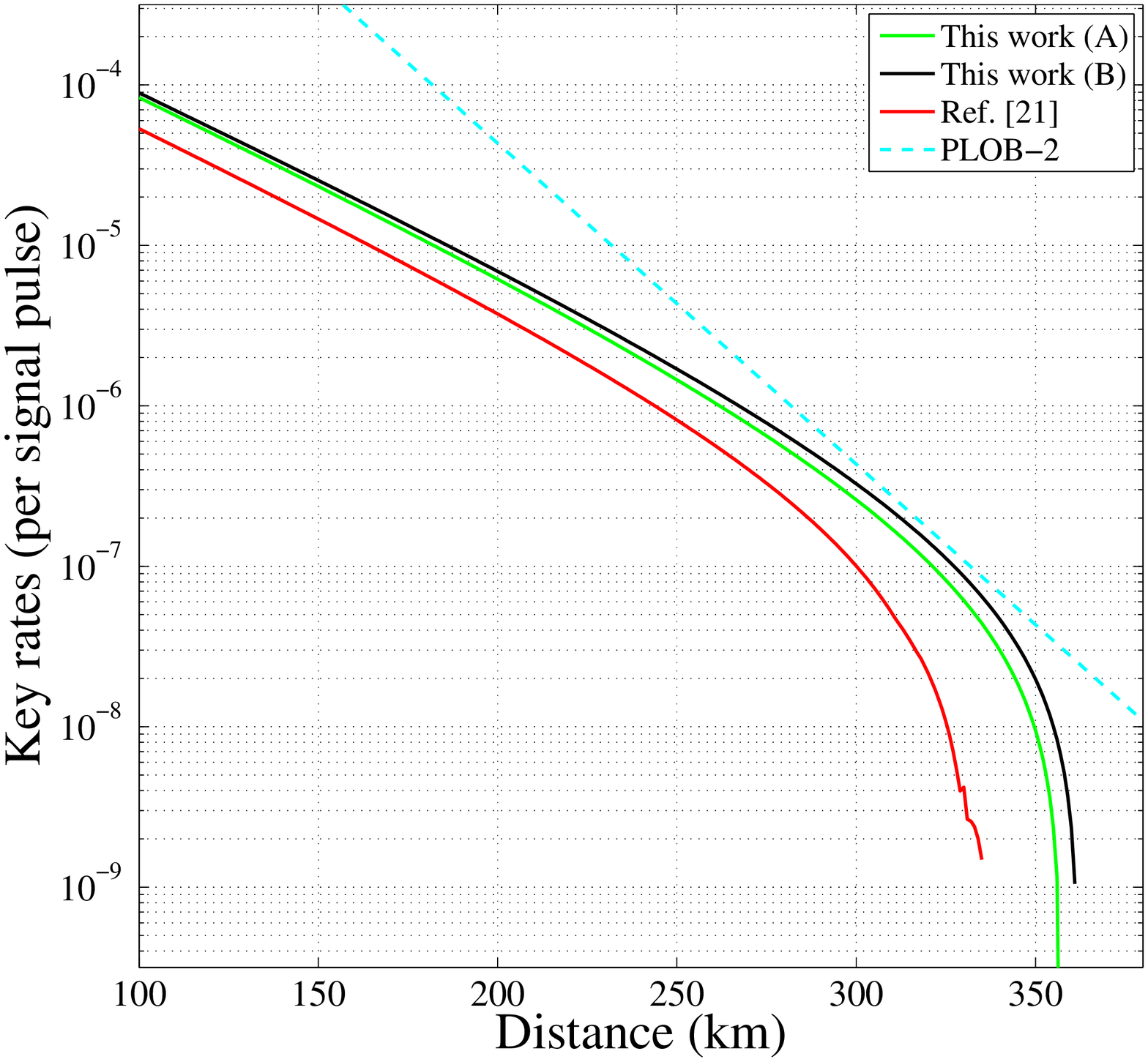}
\caption{The optimal key rates (per pulse) versus transmission distance with total number of pulses $N=10^{11}$. Here we set $L_A-L_B=100$ km. The experimental parameters that we used in the numerical simulation are listed in Table \ref{exproperty}. The black solid line is the optimized results of this work (B). The green solid line is the optimized results of this work (A). And The red solid line is the optimized results of Ref.~\cite{hu2019general}.}\label{figure4}
\end{figure}

Figure~\ref{figure1} and Figure~\ref{figure2} are our simulation results of this work and Ref.~\cite{jiang2019unconditional} with the experimental parameters list in Table.~\ref{exproperty}. In Fig.~\ref{figure1} and Fig.~\ref{figure2}, we set $L_A=L_B$ and the source parameters of Alice and Bob are all the same. The dashed brown lines in Fig.~\ref{figure1} and Fig.~\ref{figure2} are the results of absolute PLOB bound, which bounds the key rate of repeater-less QKD with whatever devices, such as perfect detection device. The cyan dashed lines are the results of practical PLOB bound, which assumes a limited detection efficiency as listed in Table~\ref{exproperty}. The results show that our method can obviously exceed the absolute PLOB bound. We set $N=1.0\times 10^{11}$ in Fig.~\ref{figure2}, which is a more practical number of total pulse in experiment, and results show that our method still obviously exceeds the absolute PLOB bound, while the original SNS protocol just exceeds the absolute PLOB bound a little. The results of applying the improved version of McDiarmid inequality to estimate the phase-flip error rate are higher about $10\%$ while $N=1.0\times 10^{12}$ and about $20\%$ while $N=1.0\times 10^{11}$ than the methods shown in Appendix A.

Figure~\ref{figure3} and Figure~\ref{figure4} are our simulation results of this work and Ref.~\cite{hu2019general} with the experimental parameters list in Table.~\ref{exproperty}. In Fig.~\ref{figure3} and Fig.~\ref{figure4}, we set $L_A-L_B=100$ km. From Fig.~\ref{figure3} and Fig.~\ref{figure4}, we can clearly see that the our method in the case of finite key size can greatly improved the key rate of SNS protocol with asymmetric channels, especially when the channel loss is large.

\begin{table}
\begin{ruledtabular}
\begin{tabular}{ccccc}
key rate&$250$ km & $390$ km & $420$ km &$440$ km\\
\hline
this work (A)&$9.52\times 10^{-6}$&$2.05\times10^{-7}$ &$6.84\times 10^{-8}$& $2.59\times 10^{-8}$\\
\hline
this work (B)&$1.02\times 10^{-5}$&$2.36\times10^{-7}$&$8.15\times 10^{-8}$& $3.26\times 10^{-8}$\\
\hline
Ref.~\cite{xu2019general}&$4.77\times 10^{-6}$&$1.02\times10^{-7}$&$3.57\times 10^{-8}$& $1.49\times 10^{-8}$\\
\hline
Ref.~\cite{jiang2019unconditional}&$5.35\times 10^{-6}$&$7.05\times10^{-8}$&$4.21\times 10^{-9}$& $0$\\
\hline
Ref.~\cite{maeda2019repeaterless}&$5.59\times 10^{-6}$&$1.16\times10^{-8}$& $0$&$0$\\
\hline
PLOB-2 &  $4.33\times 10^{-6}$  & $6.86\times10^{-9}$  &$1.72\times 10^{-9}$& $ 6.86\times10^{-10}$\\
\hline
PLOB-1 &  $1.44\times 10^{-5}$  & $2.29\times10^{-8}$  &$5.74\times 10^{-9}$& $ 2.29\times10^{-9}$
\end{tabular}
\end{ruledtabular}
\caption{The key rates of this work, Refs.~\cite{xu2019general},~\cite{jiang2019unconditional} and \cite{maeda2019repeaterless}, and the PLOB bounds. The method of Ref.~\cite{xu2019general} used here is the standard error rejection. The parameters used here are the same as that of Figure.~\ref{figure1}.}\label{absolute}
\end{table}

Table.~\ref{absolute} is the comparison of the key rates of this work, Refs.~\cite{xu2019general},~\cite{jiang2019unconditional} and \cite{maeda2019repeaterless}, and the PLOB bounds. The method of Ref.~\cite{xu2019general} used here are the standard error rejection. The parameters used here is the same as that of Figure.~\ref{figure2}. The key rates show that the method of this work improves the key rate by more than $1$ times compared with our prior work ~\cite{jiang2019unconditional}, and exceeds the results of Ref. \cite{jiang2019unconditional,maeda2019repeaterless} in all distances. Besides, with the method here, the SNS protocol can by far break the absolute key rate limit of the repeater-less QKD and even reach more than $40$ times of the practical PLOB bound and $13$ times of the absolute PLOB bound with $10^{12}$ pulses.

\begin{table}
\begin{ruledtabular}
\begin{tabular}{ccc}
key rate&$402$ km & $502$ km\\
\hline
this work (A)&$9.98\times 10^{-8}$&$4.82\times10^{-8}$\\
\hline
this work (B)&$1.07\times 10^{-7}$&$5.38\times10^{-8}$\\
\hline
Ref.~\cite{fang2019surpassing}&$1.44\times 10^{-8}$&$1.68\times10^{-9}$
\end{tabular}
\end{ruledtabular}
\caption{The key rates of this work and Ref.~\cite{fang2019surpassing}. We use the parameters of Ref.~\cite{fang2019surpassing} in calculations, e.g., the dark count rate is $p_d=3.36\times 10^{-8}$, the misalignment-error probability is $e_d=7\%$, the detection efficiency is $\eta_d=20\%$, the fiber loss is $\alpha_f=0.185$, the failure probability is $\xi=1.69\times 10^{-10}$, and the total number pulse is $N=2.0\times 10^{13}$ for the distance of $402$ km, and $p_d=1.26\times 10^{-8}$, $e_d=9.8\%$, $\eta_d=29\%$, $\alpha_f=0.162$, $\xi=1.71\times 10^{-10}$, and $N=2.0\times 10^{13}$ for the distance of $502$ km.}\label{comp}
\end{table}

Table.~\ref{comp} is the key rates of this work and Ref.~\cite{fang2019surpassing}. We use the parameters of Ref.~\cite{fang2019surpassing} in calculations, which are $p_d=3.36\times 10^{-8}$, $e_d=7\%$, $\eta_d=20\%$, $\alpha_f=0.185$, $\xi=1.69\times 10^{-10}$, and $N=2.0\times 10^{13}$ for the distance of $402$ km, and $p_d=1.26\times 10^{-8}$, $e_d=9.8\%$, $\eta_d=29\%$, $\alpha_f=0.162$, $\xi=1.71\times 10^{-10}$, and $N=2.0\times 10^{13}$ for the distance of $502$ km. The results in Table.~\ref{comp} show that the key rates of this work are more than $30$ times that of Ref.~\cite{fang2019surpassing}.

\section{Conclusion}
In this paper, we propose a zigzag approach to verify the phase-flip error of the survived bits after OPER. Based on this, we can take all the finite key effect efficiently in calculating the non-asymptotic key rate. The numerical results show that active OPER can greatly improve the key rate of SNS protocol for both the asymmetric and symmetric channels, and unconditionally break the absolute key rate limit of repeater-less quantum key distribution. Our results can directly be used to the SNS experiments.

{\bf{Acknowledgement:}} We acknowledge the financial support in part by Ministration of Science and Technology of China through The National Key Research and Development Program of China grant No. 2017YFA0303901; National Natural Science Foundation of China grant No. 11474182, 11774198 and U1738142.

\appendix

\section{The calculation method}\label{calculation}
The calculation methods of $\mean{\underline{n_{01}}}$, $\mean{\underline{n_{10}}}$, and $\mean{\overline{e_1^{ph}}}$ are similar with that in Refs~\cite{yu2019sending,jiang2019unconditional,hu2019general,xu2019general}. And the formulas of asymmetric SNS protocol are more general. We can easily get the formulas of original SNS protocol from that of asymmetric SNS protocol by setting the same source parameters of Alice and Bob, that is, drop the subscript $^\prime$ in the formulas.

To clearly show the calculation method, we denote Alice's sources $\ket{0}$, $\ket{e^{i\theta_A}\sqrt{\mu_{1}}}$ and $\ket{e^{i\theta_A^\prime}\sqrt{\mu_{2}}}$ as $o,x$, and $y$. Similarly, we denote Bob's sources $\ket{0}$, $\ket{e^{i\theta_B}\sqrt{\mu_{1}^\prime}}$, and $\ket{e^{i\theta_B^\prime}\sqrt{\mu_{2}^\prime}}$ as $o^\prime,x^\prime$, and $y^\prime$. We denote the number of pulse pairs of source $\kappa\zeta(\kappa=o,x,y;\zeta=o^\prime,x^\prime,y^\prime)$ sent out in the whole protocol as $N_{\kappa\zeta}$, and the total number of one-detector heralded events of source $\kappa\zeta$ as $n_{\kappa\zeta}$. We define the counting rate of source $\kappa\zeta$ as $S_{\kappa\zeta}=n_{\kappa\zeta}/N_{\kappa\zeta}$, and the corresponding expected value as $\mean{S_{\kappa\zeta}}$. With all those definitions, we have
\begin{equation}
\begin{split}
N_{oo^\prime}=&\{(1-p_{z})[(1-p_{z}^\prime)p_{0}p_{0}^\prime+p_{z}^\prime p_{0}(1-\epsilon{^\prime})]+p_{z}(1-p_{z}^\prime)(1-\epsilon)p_{0}^\prime\}N\\
N_{ox^\prime}=&(1-p_{z}^\prime)p_{1}^\prime[(1-p_{z})p_{0}+p_{z}(1-\epsilon)]N\\
N_{xo^\prime}=&(1-p_{z})p_{1}[(1-p_{z}^\prime)p_{0}^\prime+p_{z}^\prime(1-\epsilon{^\prime})]N\\
N_{oy^\prime}=&(1-p_{z}^\prime)(1-p_{0}^\prime-p_{1}^\prime)[(1-p_{z})p_{0}+p_{z}(1-\epsilon)]N\\
N_{yo^\prime}=&(1-p_{z})(1-p_{0}-p_{1})[(1-p_{z}^\prime)p_{0}^\prime+p_{z}^\prime(1-\epsilon{^\prime})]N
\end{split}
\end{equation}

As sources $x,y,x^\prime,y^\prime$ are phase-randomized WCS sources, they are actually the classical mixture of different photon number states~\cite{hu2019general}. Thus we can use the decoy-state method to calculate the lower bounds of the expected values of the counting rate of states $\oprod{01}{01}$ and $\oprod{10}{10}$, which are
\begin{align}
\label{s01mean}\mean{\underline{s_{01}}}&= \frac{\mu_{2}^{\prime 2}e^{\mu_{1}^\prime}\mean{S_{ox^\prime}}-\mu_{1}^{\prime 2}e^{\mu_{2}^\prime}\mean{S_{oy^\prime}}-(\mu_{2}^{\prime 2}-\mu_{1}^{\prime 2})\mean{S_{oo^\prime}}}{\mu_{2}^\prime\mu_{1}^\prime(\mu_{2}^\prime-\mu_{1}^\prime)},\\
\mean{\underline{s_{10}}}&= \frac{\mu_{2}^2e^{\mu_{1}}\mean{S_{xo^\prime}}-\mu_{1}^2e^{\mu_{2}}\mean{S_{yo^\prime}}-(\mu_{2}^2-\mu_{1}^2)\mean{S_{oo^\prime}}}{\mu_{2}\mu_{1}(\mu_{2}-\mu_{1})}.
\end{align}
Then we can get the lower bound of the expected value of the counting rate of untagged photons
\begin{equation}
\mean{\underline{s_1}}=\frac{\mu_{1}}{\mu_{1}+\mu_{1}^\prime}\mean{\underline{s_{10}}}+\frac{\mu_{1}^\prime}{\mu_{1}+\mu_{1}^\prime}\mean{\underline{s_{01}}},
\end{equation}
and
\begin{align}
\mean{n_{10}}=Np_{z}p_{z}^\prime\epsilon(1-\epsilon{^\prime})\mu_{z}e^{-\mu_{z}}\mean{\underline{s_{10}}},\\
\mean{n_{01}}=Np_{z}p_{z}^\prime\epsilon^\prime(1-\epsilon)\mu_{z}^\prime e^{-\mu_{z}^\prime}\mean{\underline{s_{01}}}.
\end{align}

The error counting rate of the $X$ windows where Alice and Bob decide to prepare the pulses with intensities $\mu_1$ and $\mu_1^\prime$, $T_{X1}$, can be used to estimate $\mean{\overline{e_1^{ph}}}$. The criterion of error events in $X$ windows are shown in Ref.~\cite{hu2019general}. We denote the number of total pulses with intensities $\mu_1$ and $\mu_1^\prime$ send out in the $X$ windows as $N_{X1}$, and the number of corresponding error events as $m_{X1}$, then we have
\begin{equation}
T_{X1}=\frac{m_{X1}}{N_{X1}}.
\end{equation}
The upper bound of the expected value of $e_1^{ph}$ is given by
\begin{equation}\label{e1}
\mean{\overline{e_1^{ph}}}=\frac{\mean{T_{X1}}-e^{-\mu_{1}-\mu_{1}^\prime}\mean{S_{oo^\prime}}/2}{e^{-\mu_{1}-\mu_{1}^\prime}(\mu_{1}+\mu_{1}^\prime)\mean{\underline{s_1}}},
\end{equation}
where $\mean{T_{X1}}$ is the expected value of $T_{X1}$.

The Eqs.\eqref{s01mean}-\eqref{e1} are represented by expected values, but the values we get in experiment are observed values. To close the gap between the expected values and observed values, we need Chernoff bound~\cite{jiang2017measurement,chernoff1952measure}. Let $\mathcal{X}$ denote the sum of $n$ independent random variables with outcomes $0$ or $1$. $\phi$ is the expected value of $\mathcal{X}$. We have
\begin{align}
\label{mul}\phi^L(\mathcal{X})=&\frac{\mathcal{X}}{1+\delta_1(\mathcal{X})},\\
\label{muu}\phi^U(\mathcal{X})=&\frac{\mathcal{X}}{1-\delta_2(\mathcal{X})},
\end{align}
where we can obtain the values of $\delta_1(\mathcal{X})$ and $\delta_2(\mathcal{X})$ by solving the following equations
\begin{align}
\label{delta1}\left(\frac{e^{\delta_1}}{(1+\delta_1)^{1+\delta_1}}\right)^{\frac{\mathcal{X}}{1+\delta_1}}&=\frac{\xi}{2},\\
\label{delta2}\left(\frac{e^{-\delta_2}}{(1-\delta_2)^{1-\delta_2}}\right)^{\frac{\mathcal{X}}{1-\delta_2}}&=\frac{\xi}{2},
\end{align}
where $\xi$ is the failure probability. Thus we have
\begin{equation}\label{sjklower}
\phi^L({N_{\alpha\beta}S_{\alpha\beta}})=N_{\alpha\beta}\mean{\underline{S}_{\alpha\beta}},\phi^U({N_{\alpha\beta}S_{\alpha\beta}})=N_{\alpha\beta}\mean{\overline{S}_{\alpha\beta}}.
\end{equation}

Besides, we can use the Chernoff bound to help us estimate their real values from their expected values. Similar to Eqs.~\eqref{mul}- \eqref{delta2}, the observed value, $\varphi$, and its expected value, $\mathcal{Y}$, satisfy
\begin{align}
\label{observ}&\varphi^U(\mathcal{Y})=[1+\delta_1^\prime(\mathcal{Y})]\mathcal{Y},\\
\label{observL}&\varphi^L(\mathcal{Y})=[1-\delta_2^\prime(\mathcal{Y})]\mathcal{Y},
\end{align}
where we can obtain the values of $\delta_1^\prime(\mathcal{Y})$ and $\delta_2^\prime(\mathcal{Y})$ by solving the following equations
\begin{align}
\left(\frac{e^{\delta_1^\prime}}{(1+\delta_1^\prime)^{1+\delta_1^\prime}}\right)^{\mathcal{Y}}&=\frac{\xi}{2},\\
\label{endd}\left(\frac{e^{-\delta_2^\prime}}{(1-\delta_2^\prime)^{1-\delta_2^\prime}}\right)^{\mathcal{Y}}&=\frac{\xi}{2}.
\end{align}

\section{The improved version of McDiarmid inequality}\label{method2}
In this part, we show how to apply the improved version of McDiarmid inequality~\cite{chau2019application} to the numerator of Eq.~\eqref{e1}.

We can rewrite the formula of $\mean{T_{X1}}$ as $\mean{T_{X1}}=\sum_{j=1}^{N_{X1}}\widetilde{W}_j/N_{X1}$ where the value of $\widetilde{W}_j$ is $1$ ($0$) if the $j$th pulse of source $xx$ causes (dose not cause) a wrong effective event. And similarly we can rewrite the formula of $\mean{S_{oo^\prime}}$ as $\mean{S_{oo^\prime}}=\sum_{j=1}^{N_{oo^\prime}}\widetilde{W}_j^\prime/N_{oo^\prime}$ where the value of $\widetilde{W}_j^\prime$ is $1$ ($0$) if the $j$th pulse of source $oo^\prime$ causes (dose not cause) an effective event. Besides, we denote $n_T=m_{X1}+n_{oo^\prime}$ and $S_T=n_T/(N_{X1}+N_{oo^\prime})$, then we have
\begin{equation}
\begin{split}
&\mean{T_{X1}}-e^{-\mu_{1}-\mu_{1}^\prime}\mean{S_{oo^\prime}}/2\\
=&\sum_{j=1}^{N_{oo^\prime}}\frac{\widetilde{W}_j^\prime}{N_{oo^\prime}}+\frac{e^{-\mu_{1}-\mu_{1}^\prime}}{2}\sum_{j=1}^{N_{oo^\prime}}\frac{\widetilde{W}_j^\prime}{N_{oo^\prime}}\\
=&\frac{S_T}{n_T}\sum_{j=1}^{n_T}W_j,
\end{split}
\end{equation}
where the value of $W_j$ is randomly $A_1$ or $A_2$, and $A_1=\frac{N_{X1}+N_{oo^\prime}}{N_{X1}}$, $A_2=-\frac{N_{X1}+N_{oo^\prime}}{2N_{oo^\prime}}e^{-\mu_{1}-\mu_{1}^\prime}$.

According to the Corollary 1 in Ref.~\cite{chau2019application}, the true value of $\sum_{j=1}^{n_T}W_j$ is larger than its observed value by $[n_T\ln{(1/\xi)}/2]^{\frac{1}{2}}(A_1-A_2)$ with probability at most $\xi$. Finally, we have
\begin{equation}
\mean{\overline{e_1^{ph}}}=\frac{T_{X1}-e^{-\mu_{1}-\mu_{1}^\prime}S_{oo^\prime}/2+\Delta}{e^{-\mu_{1}-\mu_{1}^\prime}(\mu_{1}+\mu_{1}^\prime)\mean{\underline{s_1}}},
\end{equation}
where
\begin{equation}
\Delta=\frac{S_T}{n_T}[\frac{n_T\ln{(1/\xi)}}{2}]^{\frac{1}{2}}(A_1-A_2).
\end{equation}

\section{The simulation method of $n_g,n_t^\prime,n_{ood}$ and $E^\prime$}\label{simulationmethod}
Let $n_{c0}$ be the number of effective events in the $Z$ window that  Alice decides not sending and Bob decides sending, $n_{c1}$ be the number of effective events in the $Z$ window that  Alice decides sending and Bob decides not sending, $n_v$ be the number of effective events in the $Z$ window that both Alice and Bob decides not sending, $n_d$ be the number of effective events in the $Z$ window that both Alice and Bob decides sending. We have
\begin{align*}
&n_g=\min(\frac{n_{c0}}{2}+\frac{n_{d}}{2},\frac{n_{c1}}{2}+\frac{n_{v}}{2}),\\
&n_t^\prime=\frac{n_{c0}}{n_{c0}+n_d}\frac{n_{c1}}{n_{c1}+n_v}n_g+\frac{n_{d}}{n_{c0}+n_d}\frac{n_{v}}{n_{c1}+n_v}n_g,\\
&n_{odd}=\frac{(n_{c0}+n_d)(n_{c1}+n_v)}{n_{c0}+n_{c1}+n_v+n_d},\\
&E^{\prime}=\frac{n_vn_d}{n_{c0}n_{c1}+n_vn_d}.
\end{align*}

\bibliography{refs}

\end{document}